\documentclass[conference,letterpaper]{IEEEtran}
\let\proof\relax   

\usepackage{mathpple}
\usepackage{times}
\usepackage{amsthm,xpatch}
\usepackage{amsmath,amsfonts}
\usepackage{cite}
\usepackage{amssymb}
\usepackage{dsfont}
\usepackage{graphicx, subfigure}
\usepackage{color}
\usepackage{breqn}
\usepackage{mathtools}
\usepackage{bbm}
\usepackage{latexsym}
\usepackage[ruled, linesnumbered]{algorithm2e}
\usepackage{accents}
\usepackage{rotating}
\usepackage{tikz}
\usepackage[mathscr]{euscript}
\usepackage{bm}
\usepackage[frak=lucida]{mathalpha} 
\usepackage{comment}
\usepackage{balance}
\usepackage{float}
\renewcommand{\theequation}{\upshape\arabic{equation}}

\makeatletter
\def\maketag@@@#1{\hbox{\m@th\normalfont#1}}
\makeatother

\usepackage[inline]{trackchanges}
\newtheorem{lemma}{Lemma}
\newtheorem{theorem}{Theorem}

\newtheorem{remark}{Remark}

\newtheorem*{example*}{Example}

\makeatletter
\newcommand*{\transpose}{%
  {\mathpalette\@transpose{}}%
}
\newcommand*{\@transpose}[2]{%
  \raisebox{\depth}{$\m@th#1\intercal$}%
}
\makeatother

\IEEEoverridecommandlockouts

\begin{document}

\makeatletter
\newcommand{\raisemath}[1]{\mathpalette{\raisem@th{#1}}}
\newcommand{\raisem@th}[3]{\raisebox{#1}{$#2#3$}}
\makeatother

\newcommand{\mstk}{\hspace{-0.145cm}*}

\newcommand{\mstl}{\hspace{-0.105cm}*}

\newcommand{\mstm}{\hspace{-0.175cm}*}

\newcommand{\SB}[3]{
\sum_{#2 \in #1}\biggl|\overline{X}_{#2}\biggr| #3
\biggl|\bigcap_{#2 \notin #1}\overline{X}_{#2}\biggr|
}

\newcommand{\Mod}[1]{\ (\textup{mod}\ #1)}

\newcommand{\overbar}[1]{\mkern 0mu\overline{\mkern-0mu#1\mkern-8.5mu}\mkern 6mu}

\makeatletter
\newcommand*\nss[3]{%
  \begingroup
  \setbox0\hbox{$\m@th\scriptstyle\cramped{#2}$}%
  \setbox2\hbox{$\m@th\scriptstyle#3$}%
  \dimen@=\fontdimen8\textfont3
  \multiply\dimen@ by 4             % 4x the default rule thickness
  \advance \dimen@ by \ht0
  \advance \dimen@ by -\fontdimen17\textfont2
  \@tempdima=\fontdimen5\textfont2  % x-height
  \multiply\@tempdima by 4
  \divide  \@tempdima by 5          % 80% of the x-height
  % Modifications are only necessary if the top of the subscript is not that high:
  \ifdim\dimen@<\@tempdima
    \ht0=0pt                        % don't let the subscript interfere
    \@tempdima=\fontdimen5\textfont2
    \divide\@tempdima by 4          % 25% of the x-height
    \advance \dimen@ by -\@tempdima % if >0, add to depth of superscript!
    \ifdim\dimen@>0pt
      \@tempdima=\dp2
      \advance\@tempdima by \dimen@
      \dp2=\@tempdima
    \fi
  \fi
  #1_{\box0}^{\box2}%
  \endgroup
  }
\makeatother

\makeatletter
\renewenvironment{proof}[1][\proofname]{\par
  \pushQED{\qed}%
  \normalfont \topsep6\p@\@plus6\p@\relax
  \trivlist
  \item[\hskip\labelsep
        \itshape
%    #1\@addpunct{.}]\ignorespaces% DELETED
    #1\@addpunct{:}]\ignorespaces% ADDED
}{%
  \popQED\endtrivlist\@endpefalse
}
\makeatother

\makeatletter
\newsavebox\myboxA
\newsavebox\myboxB
\newlength\mylenA

\newcommand*\xoverline[2][0.75]{%
    \sbox{\myboxA}{$\m@th#2$}%
    \setbox\myboxB\null% Phantom box
    \ht\myboxB=\ht\myboxA%
    \dp\myboxB=\dp\myboxA%
    \wd\myboxB=#1\wd\myboxA% Scale phantom
    \sbox\myboxB{$\m@th\overline{\copy\myboxB}$}%  Overlined phantom
    \setlength\mylenA{\the\wd\myboxA}%   calc width diff
    \addtolength\mylenA{-\the\wd\myboxB}%
    \ifdim\wd\myboxB<\wd\myboxA%
       \rlap{\hskip 0.5\mylenA\usebox\myboxB}{\usebox\myboxA}%
    \else
        \hskip -0.5\mylenA\rlap{\usebox\myboxA}{\hskip 0.5\mylenA\usebox\myboxB}%
    \fi}
\makeatother

\xpatchcmd{\proof}{\hskip\labelsep}{\hskip3.75\labelsep}{}{}

\pagestyle{empty}

% \fontsize{21}{28}\selectfont 

\title{A New Approach to Harnessing Side Information \\ in Multi-Server Private Information Retrieval}

\author{Ningze Wang, Anoosheh Heidarzadeh, and Alex Sprintson\thanks{Ningze Wang and Alex Sprintson are with the Department of Electrical and Computer Engineering, Texas A\&M University %, College Station, TX 77843 USA 
(E-mail: \{ningzewang,  spalex\}@tamu.edu). Anoosheh Heidarzadeh is with the Department of Electrical and Computer Engineering, Santa Clara University, 
%Santa Clara, CA 95053 USA 
(E-mail: aheidarzadeh@scu.edu).}
}

\maketitle 

\thispagestyle{empty}

\begin{abstract}
This paper presents new solutions for Private Information Retrieval (PIR) with side information. This problem is motivated by PIR settings in which a client has side information about the data held by the servers and would like to leverage this information in order to improve the download rate. 
The problem of PIR with side information has been the subject of several recent studies that presented achievability schemes as well as converses for both multi-server and single-server settings. 
However, the solutions for the multi-server settings adapted from the solutions for the single-server setting in a rather straightforward manner, 
relying on the concept of \emph{super-messages}. 
Such solutions require an exponential degree of sub-packetization (in terms of the number of messages). 

This paper makes the following contributions. First, we revisit the PIR problem with side information and present a new approach to leverage side information in the context of PIR. The key idea of our approach is a randomized algorithm to determine the linear combinations of the sub-packets that need to be recovered from each server. 
In addition, our approach takes advantage of the fact that the identity of the side information messages does not need to be kept private, and, as a result, the information retrieval scheme does not need to be symmetric.
Second, we present schemes for PIR with side information that achieve a higher rate than previously proposed solutions and require a significantly lower degree of sub-packetization (linear in the number of servers). 
Our scheme not only achieves the highest known download rate for the problem at hand but also invalidates a previously claimed converse bound on the maximum achievable download rate.
%\alex{ 
% Our scheme, to the best of our knowledge, achieves the highest known download rate for the problem at hand. Finally, our scheme invalidates a previously claimed converse bound on the maximum achievable download rate.
%}
%To the best of our knowledge, our scheme achieves the highest download rate for the problem at hand. Finally, we invalidate a previously claimed converse bound on the maximum achievable download rate. 
%We show that our scheme can be extended to the case of coded side information. 
% {\color{red} AH: I would suggest that we do not talk about coded side information. We can publish that somewhere else.}
\end{abstract}

\section{Introduction}

Private Information Retrieval (PIR) is a fundamental problem in information security and privacy. The goal of the PIR is to retrieve data from one or more remote servers while protecting the identity of the retrieved data items. 
The problem has attracted  significant attention within the computer science community since the 1990s
%The problem has attracted the significant attention of the computer science community since the 1990s 
(see, e.g., \cite{CKGS1998}) and more recently has become a very active research topic in the information theory community (see, e.g., \cite{SJ2017,SJ2016ArbitraryTIFS,TSC2019,VBU2022,SJ2018Colluding} and references therein). 
This research { has} 
produced many interesting results on several important variations of the PIR problem, including PIR in the presence of colluding servers~\cite{SJ2018Colluding,BU2019Colluding,ZX2019,YLK2020,LJJ2021,HFLH2022,ZTSP2022}, PIR from coded data~\cite{TGKHHER2017,BU18,BAWU2020}, as well as multi-message PIR~\cite{BU2018,WHS2022}. 

This paper presents new approaches for the design of multi-server PIR schemes that leverage prior side information.
The problem is inspired by PIR scenarios where a user has prior information about the { data}
%information 
stored on the servers and would like to leverage it to maximize the download rate. 
%
%The side information might originate  from the prior data that has been downloaded privately from the servers.   

The PIR setting with side information was initially studied in  single-server scenarios (refer to, e.g.,~\cite{KGHERS2017No0,LG2018,GLH2022,HKS2018,HKS2019,HKRS2019,KKHS32019, HKS2019Journal,KHSO2021,HS2021,HS2022Reuse,LJ2022,HKGRS2018}). 
 %
% HS2022LinCap
% In addition, there is a significant body of studies on several variations of single-server PIR with side information~
%\cite{KGHERS2017No0,HKGRS2018,LG2018,HKS2018,HKS2019,HKRS2019,HKS2019Journal,KKHS32019,KHSO2021,HS2021,HS2022Reuse,LJ2022,GLH2022,HS2022LinCap}. 
 %
 %
 This line of work has been extended to multi-server settings, including cache-aided PIR~\cite{T2017,WBU2018,WBU2018No2}, 
PIR with side information~\cite{KKHS12019,KGHERS2020,LG2020CISS}, 
PIR with private side information~\cite{CWJ2020,KH2021Journal,EH2024},
PIR with private coded side information~\cite{KKHS22019}, and
mutli-message PIR with private side information~\cite{SSM2018}. In PIR problems with private side information, it is required to prevent any server from identifying the side information messages. 
This requirement results in a lower rate than in scenarios where the side information messages do not need to be kept private.

The most relevant to our work are the papers by Kadhe \emph{et al.} \cite{KGHERS2020} and Li and Gastpar~\cite{LG2020CISS}. 
In particular, \cite{KGHERS2020} extends the PIR scheme developed for the single-server PIR problem for the more general multi-server setting. 
%The scheme relies on the concept of ``super-messages'' and performs well under certain divisibility conditions. 
 %
% an achievability scheme for the multi-server scenario by adapting the  
%
% presents a PIR scheme that relies on the concept of ``super-messages'' developed in the context of the single-server PIR problem with side information.  
%
Later, \cite{LG2020CISS} claimed a converse bound on the maximum achievable download rate for the multi-server PIR in the presence of side information. 

In this work, we make the following contributions. 
First, we revisit the PIR problem with side information and present a new technique to leverage side information for the construction of more efficient mutli-server PIR schemes. 
Our approach generalizes the probabilistic method due to Tian \emph{et al.} \cite{TSC2019} for settings with side information. 
The key idea is that since the identity of the side information messages does not need to be kept private, the information retrieval scheme does not need to be symmetric. 
In particular, the information requested from a server does not need to have an equal probability across different sets of demand and side information. 
This offers more flexibility in designing the retrieval scheme, enabling us to achieve higher rates compared to symmetric schemes.

%By symmetry, we mean a given query does not need to have the same probability across different problem instances, i.e., different pairs of demand message and side information messages; 

Second, we use this approach to construct an efficient randomized algorithm that determines the linear combinations of sub-packets that need to be recovered from each server. 
%To the best of our knowledge, 
Our scheme achieves the highest known download rate for this problem.
In particular, the rate of our scheme is higher than that of the scheme in~\cite{KGHERS2020} when $M+1$ does not divide $K$, 
where $M$ is the number of side information messages %available to the user 
and $K$ is the total number of messages stored at the servers. 
In addition, our scheme requires a much lower sub-packetization level of $N-1$, compared to the sub-packetization level of $N^K$ required by the existing solutions, 
where $N$ is the number of servers. 
A lower degree of sub-packetization has many practical advantages, including reduced complexity.
%each message needs to be divided into fewer sub-packets, resulting in reduced complexity and other practical advantages. 
Moreover, the rate achieved by our scheme invalidates the converse bound on the maximum achievable download rate claimed in \cite{LG2020CISS}.

\section{Problem Setup}\label{sec:SN}
We use bold-face symbols to represent random variables and regular symbols for their realizations. 
For an integer ${i\geq 1}$, we denote the set $\{1,\dots,i\}$ by $[i]$, and denote the set ${\{0\}\cup [i]}$ by $[0:i]$. 
We denote the $i$th component of a vector $\mathrm{v}$ by $\mathrm{v}(i)$, and denote its support, which refers to the index set of its nonzero components, by $\mathrm{supp}(\mathrm{v})$. 

In this work, we revisit the problem of \emph{Private Information Retrieval with Side Information (PIR-SI)}, which was initially introduced in~\cite{KGHERS2020} and subsequently studied in~\cite{LG2020CISS}. 
For completeness, the problem formulation is presented below. 

For a prime power $q$ and an integer $L\geq 1$, let $\mathbbmss{F}_q$ be a finite field of order $q$, and $\mathbbmss{F}^L_q$ be the $L$-dimensional vector space over $\mathbbmss{F}_q$. 

Consider $N>1$ non-colluding servers each of which stores an identical copy of $K>1$ messages ${\mathrm{X}_1,\dots,\mathrm{X}_K}$, where each message $\mathrm{X}_i$ is a vector in $\mathbbmss{F}^L_q$. 
That is, each message $\mathrm{X}_i$ consists of $L$ sub-packets $\mathrm{X}_{i,1},\dots,\mathrm{X}_{i,L}$, where each sub-packet $\mathrm{X}_{i,j}$ is a symbol from $\mathbbmss{F}_q$. 
We refer to $L$ as the \emph{sub-packetization level} and $q$ as the \emph{field size}.

Consider a scenario where a user wants to retrieve the message $\mathrm{X}_{\mathrm{W}}$ for some ${\mathrm{W}\in [K]}$, and  
the user has prior knowledge of $M$ messages ${\mathrm{X}_{\mathrm{S}}:=\{\mathrm{X}_i: i\in \mathrm{S}\}}$ for some $M$-subset ${\mathrm{S}\subseteq [K]\setminus \{\mathrm{W}\}}$. 
We refer to $\mathrm{X}_{\mathrm{W}}$ as the \emph{user's demand} and $\mathrm{W}$ as the \emph{demand's index}. 
In addition, we refer to $\mathrm{X}_{\mathrm{S}}$ as the \emph{user's side information}, $\mathrm{S}$ as the \emph{side information's index set}, and $M$ as the \emph{number of side information messages}. 

In line with prior work, we make the following assumptions in this work: 
(i) $\mathbf{X}_1,\dots,\mathbf{X}_K$ are independent uniform random variables over $\mathbbmss{F}^L_{q}$; 
(ii) $\mathbf{W}$ is a uniform random variable over $[K]$, and $\mathbf{S}$ is a uniform random variable over all $M$-subsets of ${[K]\setminus \{\mathrm{W}\}}$ given $\mathbf{W}=\mathrm{W}$;
(iii) The random variables $(\mathbf{W},\mathbf{S})$ and $(\mathbf{X}_{1},\dots,\mathbf{X}_K)$ are independent;
(iv) The distribution of $(\mathbf{W},\mathbf{S})$ is known in advance to all servers; and 
(v) The realization $(\mathrm{W},\mathrm{S})$ is not initially known to any server.
%\begin{itemize}
%\item[1.] $\mathbf{X}_1,\dots,\mathbf{X}_K$ are independent uniform random variables over $\mathbbmss{F}^L_{q}$. 
%That is, $H(\mathbf{X}_{\mathrm{I}})= |\mathrm{I}| B$ for $\mathrm{I}\subseteq [K]$, where $B:= L\log_2 q$ denotes the entropy of a uniform random variable over $\mathbbmss{F}^L_q$. 
%\item[2.] $\mathbf{W}$ is a uniform random variable over $[K]$, and $\mathbf{S}$ is a uniform random variable over all $M$-subsets of ${[K]\setminus \{\mathrm{W}\}}$ given $\mathbf{W}=\mathrm{W}$. 
%\item[3.] The random variables $(\mathbf{W},\mathbf{S})$ and $(\mathbf{X}_{1},\dots,\mathbf{X}_K)$ are independent. 
%\item[4.] The distribution of the random variable $(\mathbf{W},\mathbf{S})$ is known in advance to all servers. 
%\item[5.] The realization $(\mathrm{W},\mathrm{S})$ is not initially known to any server.
%\end{itemize}

For each ${n\in [N]}$, the user generates a query $\mathrm{Q}_n^{[\mathrm{W},\mathrm{S}]}$, which is a (deterministic or stochastic) function of $(\mathrm{W},\mathrm{S})$, and sends the query $\mathrm{Q}_n^{[\mathrm{W},\mathrm{S}]}$ to server $n$. 
%Since the user has no knowledge of the messages ${\{\mathrm{X}_i: i\in [K]\setminus \mathrm{S}\}}$, it is obvious that $\mathrm{Q}_n^{[\mathrm{W},\mathrm{S}]}$ cannot depend on any of these messages. 
%Conversely,  $\mathrm{Q}_n^{[\mathrm{W},\mathrm{S}]}$ may generally depend on $\mathrm{X}_{\mathrm{S}}$. 
%However, aligned with the existing literature, we consider queries that remain universally applicable across all message realizations, i.e., the query $\mathrm{Q}_n^{[\mathrm{W},\mathrm{S}]}$ is not a function of $\mathrm{X}_1,\dots,\mathrm{X}_K$. 
For each $n\in [N]$, it is required that the query $\mathrm{Q}_n^{[\mathrm{W},\mathrm{S}]}$ does not reveal any information about the demand's index $\mathrm{W}$ to server $n$. 
That is, for all $n\in [N]$, it must hold that 
$I(\mathbf{W}; \mathbf{Q}_n^{[\mathbf{W},\mathbf{S}]}) = 0$, or equivalently, 
\begin{equation*}
\mathbb{P}(\mathbf{Q}_n^{[\mathbf{W},\mathbf{S}]}=\mathrm{Q}_n^{[\mathrm{W},\mathrm{S}]}|\mathbf{W}=\mathrm{W}^{*})=\mathbb{P}(\mathbf{Q}_n^{[\mathbf{W},\mathbf{S}]}=\mathrm{Q}_n^{[\mathrm{W},\mathrm{S}]}) 
\end{equation*} for all ${\mathrm{W}^{*}\in [K]}$. 
We refer to this requirement as the \emph{privacy condition}. 
It should be noted that the privacy condition does not imply the privacy of the side information's index set $\mathrm{S}$. 
That is, the query $\mathrm{Q}_n^{[\mathrm{W},\mathrm{S}]}$ may leak some information about $\mathrm{S}$ to server $n$, while satisfying the privacy condition. 
This differs from the Private Information Retrieval with Private Side Information (PIR-PSI) setting considered in~\cite{CWJ2020,KH2021Journal,EH2024}, 
where both the demand's index $\mathrm{W}$ and the side information's index set $\mathrm{S}$ must be kept private from each server. 

Upon receiving the query $\mathrm{Q}_n^{[\mathrm{W},\mathrm{S}]}$, server $n$ generates an answer $\mathrm{A}_n^{[\mathrm{W},\mathrm{S}]}$, and sends it back to the user. 
The answer $\mathrm{A}_n^{[\mathrm{W},\mathrm{S}]}$ is a deterministic function of $\mathrm{Q}_n^{[\mathrm{W},\mathrm{S}]}$ and $\mathrm{X}_1,\dots,\mathrm{X}_K$.
That is, ${H(\mathbf{A}_n^{[\mathrm{W},\mathrm{S}]}|\mathbf{Q}_n^{[\mathrm{W},\mathrm{S}]},\mathbf{X}_1,\dots,\mathbf{X}_K)=0}$. 
The user must be able to recover their demand $\mathrm{X}_{\mathrm{W}}$ given the collection of answers
${\mathrm{A}^{[\mathrm{W},\mathrm{S}]}:=\{\mathrm{A}_n^{[\mathrm{W},\mathrm{S}]}: n\in [N]\}}$, the collection of queries ${\mathrm{Q}^{[\mathrm{W},\mathrm{S}]}:=\{\mathrm{Q}_n^{[\mathrm{W},\mathrm{S}]}: n\in [N]\}}$, the side information $\mathrm{X}_{\mathrm{S}}$, and the realizations $\mathrm{W}$ and $\mathrm{S}$. 
That is,
\[H(\mathbf{X}_{\mathrm{W}}| \mathbf{A}^{[\mathrm{W},\mathrm{S}]},\mathbf{Q}^{[\mathrm{W},\mathrm{S}]},\mathbf{X}_{\mathrm{S}})=0.\] 
We refer to this requirement as the \emph{recoverability condition}.

The PIR-SI problem is to design a protocol for generating a collection of queries $\mathrm{Q}^{[\mathrm{W},\mathrm{S}]}$ and the corresponding collection of answers $\mathrm{A}^{[\mathrm{W},\mathrm{S}]}$ for any given $(\mathrm{W},\mathrm{S})$ such that both the privacy and recoverability conditions are satisfied. 

We define the \emph{rate} of a PIR-SI scheme as the ratio of the number of bits required by the user, 
i.e., ${H(\mathbf{X}_{\mathrm{W}})=B:= L\log_2 q}$, to the expected number of bits downloaded from all servers, i.e., ${\sum_{n=1}^{N} H(\mathbf{A}^{[\mathrm{W},\mathrm{S}]}_n|\mathbf{Q}^{[\mathrm{W},\mathrm{S}]}_n)}$, where for each server, the expectation is taken over all query realizations corresponding to that server. 
Specifically, the expected number of bits downloaded from each server $n$, i.e., $H(\mathbf{A}^{[\mathrm{W},\mathrm{S}]}_n|\mathbf{Q}^{[\mathrm{W},\mathrm{S}]}_n)$, is given by 
\[\sum \mathbb{P}(\mathbf{Q}^{[\mathrm{W},\mathrm{S}]}_n = \mathrm{Q}^{[\mathrm{W},\mathrm{S}]}_n)H(\mathbf{A}^{[\mathrm{W},\mathrm{S}]}_n|\mathbf{Q}^{[\mathrm{W},\mathrm{S}]}_n = \mathrm{Q}^{[\mathrm{W},\mathrm{S}]}_n),\]
%\begin{align*}
%\sum \mathbb{P}(\mathbf{Q}^{[\mathrm{W},\mathrm{S}]}_n = \mathrm{Q}^{[\mathrm{W},\mathrm{S}]}_n)H(\mathbf{A}^{[\mathrm{W},\mathrm{S}]}_n|\mathbf{Q}^{[\mathrm{W},\mathrm{S}]}_n = \mathrm{Q}^{[\mathrm{W},\mathrm{S}]}_n),
%\end{align*} 
where the summation is over all query realizations $\mathrm{Q}^{[\mathrm{W},\mathrm{S}]}_n$.

It is important to highlight that we define rate based on the conditional entropy of each server's answer given the user's query sent to that server, i.e.,  ${H(\mathbf{A}^{[\mathrm{W},\mathrm{S}]}_n|\mathbf{Q}^{[\mathrm{W},\mathrm{S}]}_n)}$. 
This metric accurately captures the expected number of bits downloaded from server $n$, because the user's query sent to server $n$ is known to both the user and server $n$. 
In contrast, the (unconditional) entropy of the answer from each server, i.e., ${H(\mathbf{A}^{[\mathrm{W},\mathrm{S}]}_n)}$, 
which was used in previous works, e.g., in~\cite{SJ2017} and~\cite{SJ2016ArbitraryTIFS}, for rate calculation, 
takes also into account the randomness of the query. 
This can potentially introduce excessive uncertainty, which is irrelevant to the expected number of downloaded bits.
It should be noted that, regardless of which rate definition is used, the rate of the PIR-SI scheme presented in~\cite{KGHERS2020} remains unchanged. 
This is because, in these schemes, 
the length of the answer (i.e., the number of bits downloaded) from each server is the same for all query realizations. 
As a consequence, the entropy of the answer and the conditional entropy of the answer given the query are equal. 
However, this equivalence may not hold for PIR-SI schemes in which the answer length varies across different query realizations.

%A similar observation is also applicable to the definition of rate in the classical PIR problem. 
%In~\cite{SJ2017} and~\cite{SJ2016ArbitraryTIFS}, the rate of a PIR protocol is defined based on the entropy of the answer from each server, whereas in~\cite{TSC2019}, the rate of a PIR protocol is defined based on the conditional entropy of each server's answer given the user's query sent to that server. 
%Accordingly, the rate of the PIR schemes of~\cite{SJ2017} and~\cite{SJ2016ArbitraryTIFS} remains constant regardless of the rate definition used. 
%On the contrary, the PIR scheme of~\cite{TSC2019} achieves the maximum possible rate according to the latter definition, but it is not rate optimal when considering the former definition. 

In this work, our goal is to characterize the \emph{capacity} of PIR-SI, defined as the maximum achievable rate among all PIR-SI schemes, in terms of the number of servers $N$, the number of messages $K$, the number of side information messages $M$, the sub-packetization level $L$, and the field size $q$. 
We consider all values of ${N>1}$, ${K>1}$, ${1\leq M\leq K-1}$, and ${q\geq 2}$. 
Additionally, we focus on cases with ${1\leq L\leq N-1}$. 
Such small values of $L$ are particularly interesting when considering schemes that utilize linear operations at the sub-packet level and require a low sub-packetization level. 
%These schemes are of practical significance due to their advantageous characteristics in terms of computational complexity and communication overhead. 
%Among these schemes, the most appealing are arguably the \emph{one-shot} schemes, where each server's answer to the user's query comprises no more than one linear combination of the message sub-packets. 
%In such schemes, the user is limited to retrieving at most one sub-packet of the demand message from each server, and 
%hence, the sub-packetization level $L$ cannot exceed the number of servers $N$. 
%On the other hand, values of $L$ greater than $N$ become of interest for \emph{multi-shot} schemes that enable the user to recover multiple linear combinations of the message sub-packets from each server. 
%However, such schemes are beyond the scope of this work and are the subject of an ongoing research.

\section{Main Results}
In this section, we present our main results and discuss the existing bounds on the capacity of PIR-SI for all admissible values of $N$, $K$, $M$, $L$, and $q$. 
%Additionally, we review the results from prior work and provide new capacity bounds for the classical PIR problem which can be viewed as a special case of the PIR-SI problem when $M=0$.

To simplify the notation, we define
$g:=\lceil \frac{K}{M+1}\rceil$, and 
${r_k:=\prod_{j=1}^{k} \frac{K/(M+1)-j}{j}}$ for all ${k\in [g-1]}$.

%\begin{equation}\label{eq:g}
%g:=\left\lceil \frac{K}{M+1}\right\rceil,
%\end{equation} 
%and 
%\begin{equation}\label{eq:rk}
%r_k:=\prod_{j=1}^{k} \frac{K/(M+1)-j}{j}, \quad {k\in [g-1]}.
%\end{equation}

\begin{theorem}\label{thm:1}
The capacity of PIR-SI with $N$ servers, $K$ messages, $M$ side information messages, the sub-packetization level $L=N-1$, and the field size $q$, is lower bounded by 
\begin{equation}\label{eq:R}
R :=(N-1)\left(N- \left(1+\sum_{k=1}^{g-1}r_k (N-1)^k \right)^{-1}\right)^{-1}.
\end{equation}
%where $g$ and $r_k$'s are defined in~\eqref{eq:g} and~\eqref{eq:rk}, respectively. 
\end{theorem}

To establish this result, we present a novel PIR-SI scheme that achieves the rate $R$ defined in~\eqref{eq:R}. 
The retrieval process in our scheme relies on linear operations at the sub-packet level, and requires recovering no more than one linear combination of message sub-packets from each server. 
The distinctive aspect of our scheme lies in its randomized strategy to identify the message sub-packets to be combined and determine the specific linear combination to be recovered from each server. 
This strategy is carefully designed to ensure both the privacy and recoverability conditions, while taking advantage of the fact that the side information messages do not require privacy. 

%Analyzing the achievable rate $R$ yields the following result. 

\begin{theorem}\label{thm:2}
With $R$ defined in~\eqref{eq:R} and $R^{*}$ defined as
\begin{equation}\label{eq:Rstar}
R^{*}:=\frac{N^{g}-N^{g-1}}{N^{g}-1},
\end{equation} %where $g$ is defined in~\eqref{eq:g}, 
it follows that $R\geq R^{*}$ for all values of $N$, $K$, and $M$.
In particular, we have ${R=R^{*}}$ when $K$ is divisible by $M+1$, and ${R>R^{*}}$ when $K$ is not divisible by $M+1$.
%\begin{itemize}
%\item[(i)] ${R=R^{*}}$ when $K$ is divisible by $M+1$; 
%\item[(ii)] ${R>R^{*}}$ when $K$ is not divisible by $M+1$.
%\end{itemize}
\end{theorem}

\begin{theorem}\label{thm:3}
The capacity of PIR-SI with the same parameters as in Theorem~\ref{thm:1}, except when $L<N-1$, is lower bounded by ${R}_L$ defined as in~\eqref{eq:R}, except when $N$ is replaced by $L+1$.
\end{theorem}

The proof follows from a simple two-step application of our PIR-SI scheme, and is omitted for brevity. 
Here, we outline the key steps of the proof.
First, the user chooses ${L+1}$ servers arbitrarily. 
Then, the user employs the proposed scheme designed for the case with $L+1$ servers, $K$ messages, $M$ side information messages, and the sub-packetization level $L$, to privately retrieve the demand message with the help of the side information messages. 
This achieves the rate $R_L$ defined similarly to $R$ as in~\eqref{eq:R} when $N$ is replaced by $L+1$.

\begin{remark}\label{rem:1}\normalfont
Our PIR-SI scheme serves as a nontrivial extension of the PIR scheme presented in~\cite{TSC2019} to incorporate side information. 
In the PIR scheme of~\cite{TSC2019}, the space of query realizations includes a linear combination of any subset of messages, or more generally, a linear combination of their respective sub-packets with one sub-packet per message. 
However, in the PIR-SI setting, not all of these query realizations are necessary, and excluding certain query realizations can potentially increase the rate. 
Specifically, it is unnecessary to retrieve any side information message or any linear combination of any subset of side information messages. 
However, excluding only these query realizations can compromise privacy. 
%For example, 
%if retrieving any side information message directly is not allowed, the same restriction must be applied to the direct retrieval of the demand message; otherwise, the privacy condition is violated. 
This implies the need to exclude additional query realizations beyond those specifically related to side information messages. 
Moreover, in the PIR scheme of~\cite{TSC2019}, privacy is ensured by employing a uniform probability distribution over all query realizations. 
However, when some query realizations are excluded, it may be necessary to introduce a non-uniform probability distribution over the remaining ones to maintain privacy. 
Our PIR-SI scheme maximizes the rate while ensuring privacy by identifying query realizations to exclude and optimizing the probability distribution over the remaining ones.  
\end{remark}

\begin{remark}\label{rem:2}\normalfont
The best previously-known achievable rate for PIR-SI was reported in~\cite{KGHERS2020}.
The PIR-SI scheme presented in~\cite{KGHERS2020} is based on the following approach: 
(i) forming $g$ linear combinations of messages, referred to as the \emph{coded messages}, each of which is the sum of a distinct subset of messages, 
where one of these subsets contains only the demand message and a subset of the side information messages, and 
(ii) privately retrieving one of these $g$ coded messages, namely, the one containing the demand message, through the application of a capacity-achieving PIR scheme. 
As shown in~\cite{KGHERS2020}, 
%By employing the PIR scheme of~\cite{TSC2019}, which is the only existing PIR scheme achieving capacity when $L=N-1$, 
this scheme achieves the rate $R^{*}$ defined in~\eqref{eq:Rstar}.
%Comparing our achievable rate $R$ with the baseline achievable rate $R^{*}$ reveals an interesting observation. 
As evident from the result of Theorem~\ref{thm:2}, 
our scheme and the one in~\cite{KGHERS2020} achieve the same rate when $M+1$ divides $K$, yet remarkably, our scheme achieves a higher rate when $M+1$ does not divide $K$. 
The main difference between these schemes lies in determining which sub-packets of which messages must be combined in each query realization. 
In the scheme of~\cite{KGHERS2020}, the query realizations are linear combinations of only $g$ coded messages, each corresponding to a specific subset of messages. 
This is while our scheme incorporates many more subsets of messages to construct query realizations.
Moreover, the scheme of~\cite{KGHERS2020} operates exclusively on the sub-packets of the underlying coded messages, where the $l$th sub-packet of each coded message is the sum of the $l$th sub-packet of the corresponding messages. 
In contrast, our scheme goes beyond this limitation by allowing the combination of sub-packets with different indices for any subset of messages. 
While this flexibility does not result in a rate increase when $K$ is divisible by $M+1$, its advantage becomes apparent when $K$ is not divisible by $M+1$. 
Specifically, this flexibility allows for utilizing a more efficient probability distribution over query realizations, leading to an improved rate performance. 
\end{remark}

\begin{remark}\label{rem:3}\normalfont
The capacity of PIR-SI is obviously upper bounded by $1$, 
which coincides with the achievable rate $R$ when $M=K-1$. 
However, a nontrivial upper bound remains unknown for cases where $1\leq M\leq K-2$.
It was reported in~\cite{LG2020CISS} that the capacity is upper bounded by $R^{*}$.   
While $R^{*}$ matches $R$ when $K$ is divisible by $M+1$, 
$R^{*}$ is strictly less than $R$ when $K$ is not divisible by $M+1$.
This shows that $R^{*}$ is not a valid upper bound on the capacity in non-divisible cases. 
Moreover, it remains unknown whether $R^{*}$ is indeed the capacity in divisible cases or there exists a scheme that can achieve a rate higher than $R^{*}$ in such cases.  
\end{remark}

\section{Proof of Theorem~\ref{thm:1}}
In this section, we present a new PIR-SI scheme that achieves the rate $R$ defined in~\eqref{eq:R} 
for any admissible set of values of $N$, $K$, $M$, and $q$, when $L=N-1$. 
%\subsection{Proposed Scheme}

The proposed scheme consists of two steps as outlined below. 

%\vspace{0.125cm}
\textbf{Step 1:} 
Given the demand's index $\mathrm{W}$ and the side information's index set $\mathrm{S}$, the user generates 
$N$ queries $\mathrm{Q}^{[\mathrm{W},\mathrm{S}]}_1,\dots,\mathrm{Q}^{[\mathrm{W},\mathrm{S}]}_N$, and sends the query $\mathrm{Q}^{[\mathrm{W},\mathrm{S}]}_n$ to server $n$. 

The user's query to each server must specify which sub-packets of which messages are to be combined and what specific combination coefficients are to be used to generate the corresponding answer.
In the proposed scheme, each server's answer is either \emph{null}, meaning that the server does not send an answer back to the user, or composed of a \emph{sum} that includes exactly one sub-packet of every message within a designated subset of messages. 
Thus, the user's query to each server can be represented by a vector $\mathrm{v}\in {[0:N-1]^K}$, 
where an all-zero vector corresponds to a null answer, and any nonzero vector corresponds to a specific sum. 
Specifically, if $\mathrm{v}(i)$ is zero, it means that none of the sub-packets $\mathrm{X}_{i,1},\dots,\mathrm{X}_{i,N-1}$ contributes to the sum that forms the server's answer, and 
if $\mathrm{v}(i)$ is equal to $j\neq 0$, it means that the sub-packet $\mathrm{X}_{i,j}$ contributes to the sum forming the answer of the server. 

To generate the queries $\mathrm{Q}^{[\mathrm{W},\mathrm{S}]}_1,\dots,\mathrm{Q}^{[\mathrm{W},\mathrm{S}]}_N$, 
the user employs a randomized algorithm---which depends on $\mathrm{W}$ and $\mathrm{S}$---to construct $N$ vectors $\mathrm{v}_1,\dots,\mathrm{v}_N\in [0:N-1]^K$, and 
sends the vector $\tilde{\mathrm{v}}_n:=\mathrm{v}_{\pi(n)}$ to server $n$ as the query $\mathrm{Q}^{[\mathrm{W},\mathrm{S}]}_n$, 
where ${\pi:[N]\rightarrow [N]}$ is a randomly chosen permutation of $[N]$. 
The algorithm for constructing $\mathrm{v}_1,\dots,\mathrm{v}_N$ is described below.

Let $g := \lceil K/(M+1)\rceil$, and $r_k:=\prod_{j=1}^{k} \frac{K/(M+1)-j}{j}$ for all ${k\in [g-1]}$. 
First, the user randomly selects an integer, denoted by $I$, from $[0:g-1]$, where the probability of selecting each $k\in [0:g-1]$ is equal to $P_k$, where\vspace{-0.088cm} 
\[P_0 = \left(1+\sum_{k=1}^{g-1} r_k (N-1)^k \right)^{-1}\] and $P_k=P_0(N-1)^k r_k$ for all $k\in [g-1]$.
%for all $k\in [g-1]$. 
%where $r_k$'s are defined in~\eqref{eq:rk}.   
Given the value of $I$, the user then proceeds as follows:

\vspace{0.125cm}
\subsubsection{$I\neq g-1$}
First, the user randomly selects a vector ${\mathrm{a}\in [0:N-1]^{K}}$ with 
${\mathrm{supp}(\mathrm{a})\subset [K]\setminus \{\mathrm{W}\}\cup \mathrm{S}}$ and ${|\mathrm{supp}(\mathrm{a})| = I(M+1)}$.
If $I=0$, the vector $\mathrm{a}$ is all-zero. 
Otherwise, if $I\neq 0$, 
the vector $\mathrm{a}$ corresponds to a sum of ${I(M+1)}$ sub-packets, each belonging to one message from the set of interference messages, i.e., messages that are neither demand nor side information. 

Next, the user randomly selects a vector ${\mathrm{b}\in [0:N-1]^{K}}$ with $\mathrm{supp}(\mathrm{b}) = \mathrm{S}$. 
The vector $\mathrm{b}$ corresponds to a sum of $M$ sub-packets, each belonging to one side information message.  
Then, the user forms the vectors $\mathrm{v}_1,\dots,\mathrm{v}_N$ by setting 
\[\mathrm{v}_1 = \mathrm{a}, \text{ and } \mathrm{v}_{n+1} = \mathrm{a}+\mathrm{b}+\mathrm{c}_n, \quad \forall n\in [N-1],\]
where $\mathrm{c}_j\in [0:N-1]^{K}$ represents the vector corresponding to the $j$th sub-packet of the demand message $\mathrm{X}_{\mathrm{W}}$. 
That is, 
$\mathrm{c}_j(i)=0$ for $i\in [K]\setminus \{\mathrm{W}\}$, and $\mathrm{c}_j(\mathrm{W})=j$. 

\vspace{0.125cm}
\subsubsection{$I=g-1$}
First, the user randomly selects a vector ${\mathrm{a}\in [0:N-1]^{K}}$ with 
$\mathrm{supp}(\mathrm{a})= [K]\setminus \{\mathrm{W}\}\cup \mathrm{S}$.
The vector $\mathrm{a}$ corresponds to a sum of ${K-M-1}$ sub-packets, each belonging to one interference message. 

Next, the user randomly selects a vector $\mathrm{b}_1\in [0:N-1]^{K}$ with 
$\mathrm{supp}(\mathrm{b}_1) \subseteq \mathrm{S}$ and  
$|\mathrm{supp}(\mathrm{b}_1)|=g(M+1)-K$. 
The user then randomly selects a vector ${\mathrm{b}_2\in [0:N-1]^{K}}$ with ${\mathrm{supp}(\mathrm{b}_2) = \mathrm{S}\setminus \mathrm{supp}(\mathrm{b}_1)}$. 
Note that $|\mathrm{supp}(\mathrm{b}_2)|=K-(g-1)(M+1)-1$. 
The vectors $\mathrm{b}_1$ and $\mathrm{b}_2$ respectively correspond to a sum of ${g(M+1)-K}$ and ${K-(g-1)(M+1)-1}$ sub-packets, 
each belonging to one side information message.  
Then, the user forms the vectors $\mathrm{v}_1,\dots,\mathrm{v}_N$ by setting
\[\mathrm{v}_1 = \mathrm{a}+\mathrm{b}_1, \text{ and } \mathrm{v}_{n+1} = \mathrm{a}+\mathrm{b}_1+\mathrm{b}_2+\mathrm{c}_{n}, \quad \forall n\in [N-1],\]
where $\mathrm{c}_1,\dots,\mathrm{c}_{N-1}$ are defined as before.  

\vspace{0.125cm}
\textbf{Step 2:} Upon receiving the query $\mathrm{Q}^{[\mathrm{W},\mathrm{S}]}_n = \tilde{\mathrm{v}}_n$ ($=\mathrm{v}_{\pi(n)}$), 
if the vector $\tilde{\mathrm{v}}_n$ is nonzero, 
server $n$ computes the sum \[\tilde{\mathrm{Y}}_n := \sum_{i\in \mathrm{supp}(\tilde{\mathrm{v}}_n)} \mathrm{X}_{i,\tilde{\mathrm{v}}_n(i)},\] and 
sends back $\tilde{\mathrm{Y}}_n$ to the user as the answer $\mathrm{A}^{[\mathrm{W},\mathrm{S}]}_n$. 
Otherwise, if the vector $\tilde{\mathrm{v}}_n$ is all-zero, 
server $n$ does not send back any answer to the user, and 
the user sets $\tilde{\mathrm{Y}}_n = 0$. 

\begin{lemma}\label{lem:Main}
The proposed scheme satisfies both the privacy and recoverability conditions and achieves the rate $R$ in~\eqref{eq:R}.  
\end{lemma}

\begin{proof}
The proof is given in Appendix. 
\end{proof}

\section{Proof of Theorem~\ref{thm:2}}
We need to show that (i) $R = R^{*}$ when $K$ is divisible by $M+1$, and (ii) $R>R^{*}$ when $K$ is not divisible by $M+1$, where $R$ and $R^{*}$ are defined in~\eqref{eq:R} and~\eqref{eq:Rstar}, respectively.

First, consider the case (i). 
In this case, $g = K/(M+1)$. 
Thus, 
$r_k = \prod_{j=1}^{k} \frac{g-j}{j} = \binom{g-1}{k}$ for all $k\in [g-1]$. 
Moreover, by the binomial theorem, $1+\sum_{k=1}^{g-1} \binom{g-1}{k}(N-1)^k = N^{g-1}$. 
Thus, it follows that $P_0 = N^{1-g}$,
%\begin{align*}
%P_0 = \left(1+\sum_{k=1}^{g-1} \binom{g-1}{k}(N-1)^k \right)^{-1} = N^{1-g},
%\end{align*} 
which implies that 
\begin{align*}
R=\frac{N-1}{N-P_0} = \frac{N-1}{N-N^{1-g}} = \frac{N^g-N^{g-1}}{N^g-1} = R^{*}.
\end{align*}

Now, consider the case (ii). In this case, $g>K/(M+1)$. 
Thus, $r_k = \prod_{j=1}^k \frac{K/(M+1)-j}{j}<\prod_{j=1}^{k} \frac{g-j}{j} = \binom{g-1}{k}$ for all $k\in [g-1]$. 
This implies that $\sum_{k=1}^{g-1} r_k (N-1)^k < \sum_{k=1}^{g-1} \binom{g-1}{k} (N-1)^k$, and consequently,  
\begin{align*}
P_0 & = \left(1+\sum_{k=1}^{g-1} r_k (N-1)^k \right)^{-1}\\
& > \left(1+\sum_{k=1}^{g-1} \binom{g-1}{k} (N-1)^k \right)^{-1} = N^{1-g},
\end{align*} implying that 
\begin{align*}
R=\frac{N-1}{N-P_0} > \frac{N-1}{N-N^{1-g}} = \frac{N^g-N^{g-1}}{N^g-1} = R^{*}.
\end{align*}

\section{An Illustrative Example} 

Consider $N=3$ servers, each storing $K=3$ messages ${\mathrm{X}_1,\mathrm{X}_2,\mathrm{X}_3}$, where each message $\mathrm{X}_i$ consists of $L=N-1=2$ sub-packets $\mathrm{X}_{i,1}$ and $\mathrm{X}_{i,2}$. 
Suppose that a user wants to retrieve one message, say, $\mathrm{X}_1$, and the user knows $M=1$ other message, say, $\mathrm{X}_2$. 
That is, $\mathrm{W}=1$ and  $\mathrm{S}=\{2\}$. 

To generate the $N=3$ queries, one for each server, the user proceeds as follows to construct $N=3$ vectors $\mathrm{v}_1,\mathrm{v}_2,\mathrm{v}_3\in [0:N-1]^K=[0:2]^3$,  
each of length ${K=3}$ with components from ${[0:N-1]=\{0,1,2\}}$.

Noting that ${g=\lceil K/(M+1)\rceil =2}$, the user begins by randomly selecting $I\in [0:g-1]=\{0,1\}$, 
where the probability of $I=0$ (or $I=1$) is $P_0$ (or $P_1$). 
In this example, $P_0 = (1+\sum_{k=1}^{g-1} r_k (N-1)^k)^{-1} = \frac{1}{2}$ and $P_1 = 1-P_0 = \frac{1}{2}$. 

If $I=0$, the user sets $\mathrm{a} = [0,0,0]$ and randomly selects  $\mathrm{b}\in\{[0,1,0],[0,2,0]\}$. 
Note that 
(i) $\mathrm{supp}(\mathrm{a}) = \emptyset$, and 
(ii) $\mathrm{supp}(\mathrm{b}) = \mathrm{S} = \{2\}$.
Otherwise, if ${I=1}$, 
the user randomly selects ${\mathrm{a}\in\{[0,0,1],[0,0,2]\}}$ and $\mathrm{b}_1\in{\{[0,1,0],[0,2,0]\}}$, and 
sets ${\mathrm{b}_2=[0,0,0]}$. 
Note that
(i) $\mathrm{supp}(\mathrm{a}) = {[K]\setminus \{\mathrm{W}\}\cup \mathrm{S}} = \{3\}$, 
(ii) $\mathrm{supp}(\mathrm{b}_1) \subseteq \mathrm{S} = \{2\}$ and $|\mathrm{supp}(\mathrm{b}_1)| = g(M+1)-K=1$, and 
(iii) $\mathrm{supp}(\mathrm{b}_2) = {\mathrm{S}\setminus \mathrm{supp}(\mathrm{b}_1)} = \emptyset$. 

Given the vectors $\mathrm{a}$ and $\mathrm{b}$ (or $\mathrm{b}_1$ and $\mathrm{b}_2$) and the ${N-1=2}$ vectors ${\mathrm{c}_1 = [1,0,0]}$ and ${\mathrm{c}_2 = [2,0,0]}$ with $\mathrm{supp}(\mathrm{c}_1)=\mathrm{supp}(\mathrm{c}_2)=\{\mathrm{W}\} = \{1\}$, 
the user constructs the vectors
$\mathrm{v}_1,\mathrm{v}_2,\mathrm{v}_3$ 
by using the algorithm described in Step~1. 
In this example, there are six feasible sets of vectors $\{\mathrm{v}_1, \mathrm{v}_2, \mathrm{v}_3\}$ in total, with two sets for $I=0$ and the remaining four sets for $I=1$.
The table presented below enumerates all six sets of vectors $\{\mathrm{v}_1, \mathrm{v}_2, \mathrm{v}_3\}$ along with their corresponding probabilities $P_{\{\mathrm{v}_1,\mathrm{v}_2,\mathrm{v}_3\}}$, where the first two rows correspond to $I=0$, and the last four rows correspond to $I=1$.%\vspace{-0.125cm}  

{\renewcommand{\arraystretch}{1.5}
\begin{table}[h]
    \centering
    %\caption{}
    \scalebox{1}{
    \begin{tabular}{|c|c|c|c|}
    %\multicolumn{4}{c}{$\mathrm{W}=1,\mathrm{S}=\{2\}$}\\
    \hline
    $\mathrm{v}_{1}$ & $\mathrm{v}_{2}$ & $\mathrm{v}_{3}$  &   $P_{\{\mathrm{v}_1,\mathrm{v}_2,\mathrm{v}_3\}}$ \\  
    \hline
    $[0,0,0]$ & $[1,1,0]$ & $[2,1,0]$ & ${P_0}/{2}$\\ 
    \hline
    $[0,0,0]$ & $[1,2,0]$ & $[2,2,0]$ & ${P_0}/{2}$\\ 
    \hline
    $[0,1,1]$ & $[1,1,1]$ & $[2,1,1]$ & ${P_1}/{4}$\\ 
    \hline
    $[0,2,1]$ & $[1,2,1]$ & $[2,2,1]$ & ${P_1}/{4}$\\ 
    \hline
    $[0,1,2]$ & $[1,1,2]$ & $[2,1,2]$ & ${P_1}/{4}$\\ 
    \hline
    $[0,2,2]$ & $[1,2,2]$ & $[2,2,2]$ & ${P_1}/{4}$\\ 
    \hline
    \end{tabular}
    }
    \label{example}
\end{table}
}

Followed by constructing $\mathrm{v}_1,\mathrm{v}_2,\mathrm{v}_3$, 
the user randomly selects a permutation $\pi: [3]\rightarrow [3]$, and sends the query vector $\mathrm{v}_{\pi(n)}$ to server $n$.
Upon receiving ${\mathrm{v}}_{\pi(n)}$, 
if ${\mathrm{v}}_{\pi(n)}$ is nonzero, 
server $n$ sends back the answer sum ${\mathrm{Y}}_{\pi(n)} = \sum_{i\in \mathrm{supp}({\mathrm{v}}_{\pi(n)})}\mathrm{X}_{i,{\mathrm{v}}_{\pi(n)}(i)}$; 
otherwise, 
if ${\mathrm{v}}_{\pi(n)}$ is all-zero, 
no answer is sent back, and ${\mathrm{Y}_{\pi(n)}=0}$. 
In this example, there are six feasible sets of sums $\{\mathrm{Y}_1,\mathrm{Y}_2,\mathrm{Y}_3\}$, one for each set of vectors $\{\mathrm{v}_1, \mathrm{v}_2, \mathrm{v}_3\}$. 
Table~\ref{tab:1} lists all six sets of sums $\{\mathrm{Y}_1,\mathrm{Y}_2,\mathrm{Y}_3\}$ and their probabilities $P_{\{\mathrm{Y}_1,\mathrm{Y}_2,\mathrm{Y}_3\}}$. 
For completeness, we also present Table~\ref{tab:2} which lists all six sets of sums $\{\mathrm{Y}_1,\mathrm{Y}_2,\mathrm{Y}_3\}$ and their probabilities $P_{\{\mathrm{Y}_1,\mathrm{Y}_2,\mathrm{Y}_3\}}$, for the general case of $\mathrm{W} = i$, $\mathrm{S} = \{j\}$, and $[K]\setminus \{\mathrm{W}\}\cup \mathrm{S} = [3]\setminus \{i,j\} = \{k\}$, where $\{i,j,k\}$ is an arbitrary permutation of $[3]$. 

{\renewcommand{\arraystretch}{1.25}
\begin{table}[t]
    \centering
    \caption{$\mathrm{W}=1,\mathrm{S}=\{2\}$\vspace{-0.25cm}}\label{tab:1}
    \scalebox{1}{
    \begin{tabular}{|c|c|c|c|}
    %\multicolumn{4}{c}{$\mathrm{W}=1,\mathrm{S}=\{2\}$}\\
    \hline
    $\mathrm{Y}_{1}$ & $\mathrm{Y}_{2}$ & $\mathrm{Y}_{3}$  &   $P_{\{\mathrm{Y}_1,\mathrm{Y}_2,\mathrm{Y}_3\}}$ \\  
    \hline
    $0$ & $\mathrm{X}_{1,1}+\mathrm{X}_{2,1}$ & $\mathrm{X}_{1,2}+\mathrm{X}_{2,1}$ & ${P_0}/{2}$\\ 
    \hline
    $0$ & $\mathrm{X}_{1,1}+\mathrm{X}_{2,2}$ & $\mathrm{X}_{1,2}+\mathrm{X}_{2,2}$ & ${P_0}/{2}$\\ 
    \hline
    $\mathrm{X}_{2,1}+\mathrm{X}_{3,1}$ & $\mathrm{X}_{1,1}+\mathrm{X}_{2,1}+\mathrm{X}_{3,1}$ & $\mathrm{X}_{1,2}+\mathrm{X}_{2,1}+\mathrm{X}_{3,1}$ & ${P_1}/{4}$\\ 
    \hline
    $\mathrm{X}_{2,2}+\mathrm{X}_{3,1}$ & $\mathrm{X}_{1,1}+\mathrm{X}_{2,2}+\mathrm{X}_{3,1}$ & $\mathrm{X}_{1,2}+\mathrm{X}_{2,2}+\mathrm{X}_{3,1}$ & ${P_1}/{4}$\\ 
    \hline
    $\mathrm{X}_{2,1}+\mathrm{X}_{3,2}$ & $\mathrm{X}_{1,1}+\mathrm{X}_{2,1}+\mathrm{X}_{3,2}$ & $\mathrm{X}_{1,2}+\mathrm{X}_{2,1}+\mathrm{X}_{3,2}$ & ${P_1}/{4}$\\ 
    \hline
    $\mathrm{X}_{2,2}+\mathrm{X}_{3,2}$ & $\mathrm{X}_{1,1}+\mathrm{X}_{2,2}+\mathrm{X}_{3,2}$ & $\mathrm{X}_{1,2}+\mathrm{X}_{2,2}+\mathrm{X}_{3,2}$ & ${P_1}/{4}$\\ 
    \hline
    \end{tabular}
    }\vspace{-0.15cm}
    %\label{example}
\end{table}
}

%The following tables correspond to different instances of $\mathrm{W}$ and $\mathrm{S}$ and are included for completeness. 

{\renewcommand{\arraystretch}{1.25}
\begin{table}[t]
    \centering
    \caption{$\mathrm{W}=i,\mathrm{S}=\{j\}$\vspace{-0.25cm}}\label{tab:2}
    \scalebox{1}{
    \begin{tabular}{|c|c|c|c|}
    %\multicolumn{4}{c}{$\mathrm{W}=1,\mathrm{S}=\{3\}$}\\
    \hline
    $\mathrm{Y}_{1}$ & $\mathrm{Y}_{2}$ & $\mathrm{Y}_{3}$  &   $P_{\{\mathrm{Y}_1,\mathrm{Y}_2,\mathrm{Y}_3\}}$ \\  
    \hline
    $0$ & $\mathrm{X}_{i,1}+\mathrm{X}_{k,1}$ & $\mathrm{X}_{i,2}+\mathrm{X}_{k,1}$ & ${P_0}/{2}$\\ 
    \hline
    $0$ & $\mathrm{X}_{i,1}+\mathrm{X}_{k,2}$ & $\mathrm{X}_{i,2}+\mathrm{X}_{k,2}$ & ${P_0}/{2}$\\ 
    \hline
    $\mathrm{X}_{j,1}+\mathrm{X}_{k,1}$ & $\mathrm{X}_{i,1}+\mathrm{X}_{j,1}+\mathrm{X}_{k,1}$ & $\mathrm{X}_{i,2}+\mathrm{X}_{j,1}+\mathrm{X}_{k,1}$ & ${P_1}/{4}$\\ 
    \hline
    $\mathrm{X}_{j,2}+\mathrm{X}_{k,1}$ & $\mathrm{X}_{i,1}+\mathrm{X}_{j,2}+\mathrm{X}_{k,1}$ & $\mathrm{X}_{i,2}+\mathrm{X}_{j,2}+\mathrm{X}_{k,1}$ & ${P_1}/{4}$\\ 
    \hline
    $\mathrm{X}_{j,1}+\mathrm{X}_{k,2}$ & $\mathrm{X}_{i,1}+\mathrm{X}_{j,1}+\mathrm{X}_{k,2}$ & $\mathrm{X}_{i,2}+\mathrm{X}_{j,1}+\mathrm{X}_{k,2}$ & ${P_1}/{4}$\\ 
    \hline
    $\mathrm{X}_{j,2}+\mathrm{X}_{k,2}$ & $\mathrm{X}_{i,1}+\mathrm{X}_{j,2}+\mathrm{X}_{k,2}$ & $\mathrm{X}_{i,2}+\mathrm{X}_{j,2}+\mathrm{X}_{k,2}$ & ${P_1}/{4}$\\ 
    \hline
    \end{tabular}
    }\vspace{-0.3cm}
    %\label{example}
\end{table}
}

Consider an arbitrary $n\in [3]$, and suppose that server $n$ receives the query vector $\mathrm{v}^{*}$ from the user. 
To verify that the privacy condition is met, we need to show that ${\mathbb{P}(\mathbf{Q}_n = \mathrm{v}^{*}|\mathbf{W} = \mathrm{W}^{*}) = \mathbb{P}(\mathbf{Q}_n = \mathrm{v}^{*})}$ for all $\mathrm{W}^{*} \in [3]$.
As an example, consider the case of $\mathrm{v}^{*} = [0,2,1]$. 
In this case, server $n$ sends the answer $\mathrm{Y}^{*} = \mathrm{X}_{2,2}+\mathrm{X}_{3,1}$ to the user. 

First, consider $\mathrm{W}^{*}=1$. 
From the table corresponding to the case of $(\mathrm{W},\mathrm{S})=(1,\{2\})$, 
it follows that ${\mathbb{P}(\mathbf{Q}_n = \mathrm{v}^{*}|\mathbf{W} = 1,\mathbf{S}=\{2\})} = \frac{1}{3}\times\frac{P_1}{4} = \frac{1}{24}$. 
Similarly, in the table for the case of $(\mathrm{W},\mathrm{S})=(1,\{3\})$, 
we observe that ${\mathbb{P}(\mathbf{Q}_n = \mathrm{v}^{*}|\mathbf{W} = 1,\mathbf{S}=\{3\})} = \frac{1}{3}\times\frac{P_1}{4} = \frac{1}{24}$.
Since ${\mathbb{P}(\mathbf{S} = \{2\}|\mathbf{W}=1)}= {\mathbb{P}(\mathbf{S} = \{3\}|\mathbf{W}=1)} = \frac{1}{2}$, it follows that ${\mathbb{P}(\mathbf{Q}_n = \mathrm{v}^{*}|\mathbf{W} = 1)} = \frac{1}{2}\times \frac{1}{24}+\frac{1}{2}\times \frac{1}{24} = \frac{1}{24}$. 

Next, consider $\mathrm{W}^{*}=2$. 
From the tables corresponding to the cases of $(\mathrm{W},\mathrm{S})=(2,\{1\})$ and $(\mathrm{W},\mathrm{S})=(2,\{3\})$, it follows that 
${\mathbb{P}(\mathbf{Q}_n = \mathrm{v}^{*}|\mathbf{W} = 2,\mathbf{S}=\{1\})} = 0$ and 
${\mathbb{P}(\mathbf{Q}_n = \mathrm{v}^{*}|\mathbf{W} = 2,\mathbf{S}=\{3\})} = \frac{1}{3}\times\frac{P_0}{2} = \frac{1}{12}$. 
Since ${\mathbb{P}(\mathbf{S} = \{1\}|\mathbf{W}=2)}= {\mathbb{P}(\mathbf{S} = \{3\}|\mathbf{W}=2)} = \frac{1}{2}$, 
we find that ${\mathbb{P}(\mathbf{Q}_n = \mathrm{v}^{*}|\mathbf{W} = 2)} = \frac{1}{2}\times 0+\frac{1}{2}\times \frac{1}{12} = \frac{1}{24}$. 

Lastly, consider $\mathrm{W}^{*}=3$. 
In the tables for the cases of $(\mathrm{W},\mathrm{S})=(3,\{1\})$ and $(\mathrm{W},\mathrm{S})=(3,\{2\})$, we observe that ${\mathbb{P}(\mathbf{Q}_n = \mathrm{v}^{*}|\mathbf{W} = 3,\mathbf{S}=\{1\})} = 0$ and 
${\mathbb{P}(\mathbf{Q}_n = \mathrm{v}^{*}|\mathbf{W} = 3,\mathbf{S}=\{2\})} = \frac{1}{3}\times\frac{P_0}{2} = \frac{1}{12}$. 
Since ${\mathbb{P}(\mathbf{S} = \{1\}|\mathbf{W}=3)}= {\mathbb{P}(\mathbf{S} = \{2\}|\mathbf{W}=3)} = \frac{1}{2}$, we find that ${\mathbb{P}(\mathbf{Q}_n = \mathrm{v}^{*}|\mathbf{W} = 3)} = \frac{1}{2}\times 0+\frac{1}{2}\times \frac{1}{12} = \frac{1}{24}$.  

By the above results, $\mathbb{P}(\mathbf{Q}_n = \mathrm{v}^{*}|\mathbf{W} = \mathrm{W}^{*})=\frac{1}{24}$ for all $\mathrm{W}^{*}\in [3]$. 
Since $\mathbb{P}(\mathbf{W} = \mathrm{W}^{*}) = \frac{1}{3}$ for all $\mathrm{W}^{*}\in [3]$, it follows that ${\mathbb{P}(\mathbf{Q}_n = \mathrm{v}^{*})} = 3\times \frac{1}{3}\times \frac{1}{24} = \frac{1}{24}$. 
Thus, ${\mathbb{P}(\mathbf{Q}_n = \mathrm{v}^{*}|\mathbf{W} = \mathrm{W}^{*}) = \mathbb{P}(\mathbf{Q}_n = \mathrm{v}^{*})}$ for all $\mathrm{W}^{*} \in [3]$, implying that the privacy condition is met for the case of $\mathrm{v}^{*} = [0,2,1]$. 
Similarly, it can be shown that the privacy condition is met for all other cases of $\mathrm{v}^{*}$.

For this example, the proposed scheme achieves the rate $R = {(N-1)}/{(N-P_0)} = 4/5$, 
which is higher than the rate $R^{*} = {(N^g-N^{g-1})}/{(N^g-1)} = 3/4$ achieved by the scheme presented in~\cite{KGHERS2020}. 

\cleardoublepage

%\balance

\bibliographystyle{IEEEtran}
\bibliography{PIR_PC_Refs}

\newpage

\appendix

%\appendices

%\section{Proof of Lemma~\ref{lem:Main}}

In this section, we present the proof of Lemma~\ref{lem:Main}. 
In particular, we provide the proofs for privacy, recoverability, and the achievable rate of the proposed scheme. 

\subsection{Proof of Privacy}\label{subsec:Privacy}
We show that the chosen probabilities $P_0,P_1,\dots,P_{g-1}$ guarantee that the privacy condition is satisfied. 
To simplify the notation, we denote $\mathrm{Q}_n^{[\mathrm{W},\mathrm{S}]}$ and $\mathbf{Q}_n^{[\mathbf{W},\mathbf{S}]}$ by $\mathrm{Q}_n$ and $\mathbf{Q}_n$, respectively.

Suppose that the user sends the vector $\mathrm{v}^{*}$ to server $n$ as the query, 
i.e., ${\mathrm{Q}_n=\mathrm{v}^{*}}$. 
To satisfy the privacy condition, 
it must hold that ${\mathbb{P}(\mathbf{Q}_n=\mathrm{v}^{*}|\mathbf{W}=\mathrm{W}^{*})}=\mathbb{P}(\mathbf{Q}_n=\mathrm{v}^{*})$ for all $\mathrm{W}^{*}\in [K]$. 
Equivalently, $\mathbb{P}(\mathbf{Q}_n=\mathrm{v}^{*}|\mathbf{W}=\mathrm{W}^{*})$ must remain constant for all $\mathrm{W}^{*}\in [K]$. 

Note that $|\mathrm{supp}(\mathrm{v}^{*})|$ is either $0$, or $K$, or $k(M+1)$ for some $k\in [g-1]$. 
When $|\mathrm{supp}(\mathrm{v}^{*})|=0$ (i.e., the vector $\mathrm{v}^{*}$ is all-zero), 
it is easy to see that \[{\mathbb{P}(\mathbf{Q}_n=\mathrm{v}^{*}|\mathbf{W}=\mathrm{W}^{*}) = P_0\times \frac{1}{N}}\] for all $\mathrm{W}^{*}\in [K]$.
This is because the vector $\mathrm{v}^{*}$ is all-zero iff 
$I=0$ is selected (and hence, the vector $\mathrm{a}$ is all-zero) and 
the user sends the vector $\mathrm{v}_1 = \mathrm{a}$ to server $n$.  
Similarly, when $|\mathrm{supp}(\mathrm{v}^{*})|=K$, it readily follows that \[\mathbb{P}(\mathbf{Q}_n=\mathrm{v}^{*}|\mathbf{W}=\mathrm{W}^{*}) = P_{g-1}\times \frac{1}{(N-1)^{K-1}}\times \frac{1}{N}\] for all $\mathrm{W}^{*}\in [K]$. 
This is because $|\mathrm{supp}(\mathrm{v}^{*})|=K$ iff 
$I=g-1$ is selected, 
the $i$th component of the selected vector $\mathrm{a}$ with $\mathrm{supp}(\mathrm{a})=[K]\setminus \{\mathrm{W}^{*}\}\cup \mathrm{S}$ is equal to $\mathrm{v}^{*}(i)$ for all ${i\in [K]\setminus \{\mathrm{W}^{*}\}\cup \mathrm{S}}$, 
the $i$th component of the selected vector $\mathrm{b}_1+\mathrm{b}_2$ with $\mathrm{supp}(\mathrm{b}_1+\mathrm{b}_2)=\mathrm{S}$ is equal to $\mathrm{v}^{*}(i)$ for all ${i\in  \mathrm{S}}$, and 
the user sends the vector $\mathrm{v}_{j+1} = \mathrm{a}+\mathrm{b}_1+\mathrm{b}_2+\mathrm{c}_j$ to server $n$, where $j = \mathrm{v}^{*}(\mathrm{W}^{*})$. 
By these results, 
$\mathbb{P}(\mathbf{Q}_n=\mathrm{v}^{*}|\mathbf{W}=\mathrm{W}^{*})$ is constant for all $\mathrm{W}^{*}\in [K]$, when $|\mathrm{supp}(\mathrm{v}^{*})|=0$ or $|\mathrm{supp}(\mathrm{v}^{*})|=K$, regardless of the choice of probabilities $P_0,P_1,\dots,P_{g-1}$. 

Next, suppose that ${|\mathrm{supp}(\mathrm{v}^{*})| = k(M+1)}$ for some ${k\in [g-1]}$.
Fix an arbitrary ${\mathrm{W}^{*}\in [K]}$. 
It is easy to see that 
$\mathbb{P}(\mathbf{S}=\mathrm{S}^{*}|\mathbf{W}=\mathrm{W}^{*})=1/\binom{K-1}{M}$ for all $M$-subsets ${\mathrm{S}^{*}\subseteq [K]\setminus \{\mathrm{W}^{*}\}}$. 
Thus, $\mathbb{P}(\mathbf{Q}_n=\mathrm{v}^{*}|\mathbf{W}=\mathrm{W}^{*})$ is given by
\begin{align*}
\frac{1}{\binom{K-1}{M}} \sum_{\mathrm{S}^{*}\subseteq [K]\setminus \{\mathrm{W}^{*}\}} \mathbb{P}(\mathbf{Q}_n=\mathrm{v}^{*}|\mathbf{W}=\mathrm{W}^{*},\mathbf{S}=\mathrm{S}^{*}).
\end{align*}

We shall consider the following two cases separately: 
(i)~${\mathrm{W}^{*}\in \mathrm{supp}(\mathrm{v}^{*})}$, and 
(ii) ${\mathrm{W}^{*}\not\in \mathrm{supp}(\mathrm{v}^{*})}$.

First, consider the case (i). 
In this case, ${\mathrm{W}^{*}\in \mathrm{supp}(\mathrm{v}^{*})}$, and hence, 
$\mathbb{P}(\mathbf{Q}_n=\mathrm{v}^{*}|\mathbf{W}=\mathrm{W}^{*},\mathbf{S}=\mathrm{S}^{*})=0$ for every ${\mathrm{S}^{*}\not\subset \mathrm{supp}(\mathrm{v}^{*})}$. 
This readily follows from the query construction. 
Thus, $\mathbb{P}(\mathbf{Q}_n=\mathrm{v}^{*}|\mathbf{W}=\mathrm{W}^{*})$ is given by
\begin{align*}
& \frac{1}{\binom{K-1}{M}} \sum_{\mathrm{S}^{*}\subseteq \mathrm{supp}(\mathrm{v}^{*})\setminus \{\mathrm{W}^{*}\}} \mathbb{P}(\mathbf{Q}_n=\mathrm{v}^{*}|\mathbf{W}=\mathrm{W}^{*},\mathbf{S}=\mathrm{S}^{*}).
\end{align*}
In addition, for every $M$-subset  $\mathrm{S}^{*}\subseteq \mathrm{supp}(\mathrm{v}^{*})\setminus \{\mathrm{W}^{*}\}$, it follows that $\mathbb{P}(\mathbf{Q}_n=\mathrm{v}^{*}|\mathbf{W}=\mathrm{W}^{*},\mathbf{S}=\mathrm{S}^{*})$ is given by
\begin{align*}
& P_{k-1} \times \frac{1}{\binom{K-M-1}{(k-1)(M+1)}} \\ 
& \quad \times \frac{1}{(N-1)^{(k-1)(M+1)}} \times \frac{1}{(N-1)^M} \times\frac{1}{N}
\end{align*} 
because 
$|\mathrm{supp}(\mathrm{v}^{*})|=k(M+1)$ ($<K$ for all $k\in [g-1]$) and 
$\{\mathrm{W}^{*}\}\cup \mathrm{S}^{*}\subseteq \mathrm{supp}(\mathrm{v}^{*})$ iff 
$I=k-1$ is selected, 
the $i$th component of the selected vector $\mathrm{a}$ with $\mathrm{supp}(\mathrm{a})\subset [K]\setminus \{\mathrm{W}^{*}\}\cup \mathrm{S}$ and $|\mathrm{supp}(\mathrm{a})|=(k-1)(M+1)$ is equal to $\mathrm{v}^{*}(i)$ for all ${i\in [K]\setminus \{\mathrm{W}^{*}\}\cup \mathrm{S}}$, 
the $i$th component of the selected vector $\mathrm{b}$ with $\mathrm{supp}(\mathrm{b})=\mathrm{S}$ is equal to $\mathrm{v}^{*}(i)$ for all ${i\in  \mathrm{S}}$, and 
the user sends the vector $\mathrm{v}_{j+1} = \mathrm{a}+\mathrm{b}+\mathrm{c}_j$ to server $n$, where $j = \mathrm{v}^{*}(\mathrm{W}^{*})$.

Since $|\mathrm{supp}(\mathrm{v}^{*})| = k(M+1)$, it readily follows that there are $\binom{k(M+1)-1}{M}$ $M$-subsets $\mathrm{S}^{*}\subseteq \mathrm{supp}(\mathrm{v}^{*})\setminus \{\mathrm{W}^{*}\}$. 
Thus, ${\mathbb{P}(\mathbf{Q}_n=\mathrm{v}^{*}|\mathbf{W}=\mathrm{W}^{*})}$ is given by
\begin{align}\label{eq:PWin}
& \frac{\binom{k(M+1)-1}{M}}{\binom{K-1}{M}}\times P_{k-1} \times \frac{1}{\binom{K-M-1}{(k-1)(M+1)}} \nonumber \\
& \quad \times \frac{1}{(N-1)^{(k-1)(M+1)}}\times \frac{1}{(N-1)^M} \times\frac{1}{N}.
\end{align}

Now, consider the case (ii). 
In this case, ${\mathrm{W}^{*}\not\in \mathrm{supp}(\mathrm{v}^{*})}$. 
First, suppose that $|\mathrm{supp}(\mathrm{v}^{*})|\neq (g-1)(M+1)$ (i.e., $k\neq g-1$). 
From the query construction, it follows that 
${\mathbb{P}(\mathbf{Q}_n=\mathrm{v}^{*}|\mathbf{W}=\mathrm{W}^{*},\mathbf{S}=\mathrm{S}^{*})}=0$ 
for every $M$-subset ${\mathrm{S}^{*}\subseteq [K]\setminus \{\mathrm{W}^{*}\}}$ 
such that $|\mathrm{S}^{*}\cap \mathrm{supp}(\mathrm{v}^{*})|\neq 0$. 
On the contrary, for every $M$-subset ${\mathrm{S}^{*}\subseteq [K]\setminus \{\mathrm{W}^{*}\}}$ 
such that $|\mathrm{S}^{*}\cap \mathrm{supp}(\mathrm{v}^{*})|=0$, it should not be hard to see that 
$\mathbb{P}(\mathbf{Q}_n=\mathrm{v}^{*}|\mathbf{W}=\mathrm{W}^{*},\mathbf{S}=\mathrm{S}^{*})$ is given by 
\begin{align*}
P_{k} \times \frac{1}{\binom{K-M-1}{k(M+1)}} \times \frac{1}{(N-1)^{k(M+1)}}\times\frac{1}{N}
\end{align*}
because $|\mathrm{supp}(\mathrm{v}^{*})|=k(M+1)$ ($\neq (g-1)(M+1)$) and 
$|\mathrm{S}^{*}\cap \mathrm{supp}(\mathrm{v}^{*})|=0$ iff $I=k$ ($\neq g-1$) is selected, 
the selected vector $\mathrm{a}$ with $\mathrm{supp}(\mathrm{a})\subset [K]\setminus \{\mathrm{W}^{*}\}\cup \mathrm{S}$ and $|\mathrm{supp}(\mathrm{a})|=k(M+1)$ is equal to the vector $\mathrm{v}^{*}$, 
and the user sends the vector $\mathrm{v}_1 = \mathrm{a}$ to server $n$. 

Since $|\mathrm{supp}(\mathrm{v}^{*})| = k(M+1)$, 
there are $\binom{K-1-k(M+1)}{M}$ $M$-subsets $\mathrm{S}^{*}\subseteq [K]\setminus \{\mathrm{W}^{*}\}$ such that $|\mathrm{S}^{*}\cap \mathrm{supp}(\mathrm{v}^{*})|=0$. 
Thus, $\mathbb{P}(\mathbf{Q}_n=\mathrm{v}^{*}|\mathbf{W}=\mathrm{W}^{*})$ is given by
\begin{align}\label{eq:PWnin1}
& \frac{\binom{K-1-k(M+1)}{M}}{\binom{K-1}{M}}\times P_{k} \nonumber \\ & \quad \times \frac{1}{\binom{K-M-1}{k(M+1)}} \times \frac{1}{(N-1)^{k(M+1)}}\times\frac{1}{N}.
\end{align}

Now, suppose that ${|\mathrm{supp}(\mathrm{v}^{*})|}={(g-1)(M+1)}$ (i.e., ${k=g-1}$). 
It readily follows from the query construction that ${\mathbb{P}(\mathbf{Q}_n=\mathrm{v}^{*}|\mathbf{W}=\mathrm{W}^{*},\mathbf{S}=\mathrm{S}^{*})}=0$ 
for every $M$-subset $\mathrm{S}^{*}\subseteq [K]\setminus \{\mathrm{W}^{*}\}$ 
such that $|\mathrm{S}^{*}\cap \mathrm{supp}(\mathrm{v}^{*})|\neq {g(M+1)-K}$. 
For every $M$-subset $\mathrm{S}^{*}\subseteq [K]\setminus \{\mathrm{W}^{*}\}$ 
such that $|\mathrm{S}^{*}\cap \mathrm{supp}(\mathrm{v}^{*})|= {g(M+1)-K}$, it follows that 
$\mathbb{P}(\mathbf{Q}_n=\mathrm{v}^{*}|\mathbf{W}=\mathrm{W}^{*},\mathbf{S}=\mathrm{S}^{*})$ is given by 
\begin{align*}
& P_{g-1} \times \frac{1}{(N-1)^{K-M-1}} \\
& \quad \times \frac{1}{\binom{M}{g(M+1)-K}} \times \frac{1}{(N-1)^{g(M+1)-K}}\times\frac{1}{N}.
\end{align*}
This is because $|\mathrm{supp}(\mathrm{v}^{*})|=(g-1)(M+1)$ and 
${|\mathrm{S}^{*}\cap \mathrm{supp}(\mathrm{v}^{*})|} =g(M+1)-K$ iff $I=g-1$ is selected, 
the $i$th component of the selected vector $\mathrm{a}$ with $\mathrm{supp}(\mathrm{a})=[K]\setminus \{\mathrm{W}^{*}\}\cup \mathrm{S}$ is equal to $\mathrm{v}^{*}(i)$ for all ${i\in [K]\setminus \{\mathrm{W}^{*}\}\cup \mathrm{S}}$, 
the $i$th component of the selected vector $\mathrm{b}_1$ with $\mathrm{supp}(\mathrm{b}_1)\subseteq \mathrm{S}$ and $|\mathrm{supp}(\mathrm{b}_1)|=g(M+1)-K$ is equal to $\mathrm{v}^{*}(i)$ for all ${i\in  \mathrm{S}}$, and 
the user sends the vector $\mathrm{v}_1 = \mathrm{a}+\mathrm{b}_1$ to server $n$. 

Since $|\mathrm{supp}(\mathrm{v}^{*})| = (g-1)(M+1)$, 
there are $\binom{(g-1)(M+1)}{g(M+1)-K}$ $M$-subsets $\mathrm{S}^{*}\subseteq [K]\setminus \{\mathrm{W}^{*}\}$ such that ${|\mathrm{S}^{*}\cap \mathrm{supp}(\mathrm{v}^{*})|}=g(M+1)-K$. 
Thus, it follows that $\mathbb{P}(\mathbf{Q}_n=\mathrm{v}^{*}|\mathbf{W}=\mathrm{W}^{*})$ is given by
\begin{align}\label{eq:PWnin2}
& \frac{\binom{(g-1)(M+1)}{g(M+1)-K}}{\binom{K-1}{M}} \times P_{g-1} \times \frac{1}{(N-1)^{K-M-1}} \nonumber \\
& \quad \times \frac{1}{\binom{M}{g(M+1)-K}} \times \frac{1}{(N-1)^{g(M+1)-K}}\times\frac{1}{N}.
\end{align}

By equating~\eqref{eq:PWin} and~\eqref{eq:PWnin1} and simplifying, 
it follows that when $|\mathrm{supp}(\mathrm{v}^{*})|\in \{M+1,\dots,(g-2)(M+1)\}$, ${\mathbb{P}(\mathbf{Q}_n=\mathrm{v}^{*}|\mathbf{W}=\mathrm{W}^{*})}$ is constant for all $\mathrm{W}^{*}\in [K]$, 
if
\begin{align}\label{eq:Cond1}
P_k = (N-1)P_{k-1}\frac{\binom{k(M+1)-1}{M} \binom{K-M-1}{k(M+1)}}{\binom{K-M-1}{(k-1)(M+1)} \binom{K-1-k(M+1)}{M}}
\end{align} for all $k\in [g-2]$. 
Simplifying~\eqref{eq:Cond1} yields  
\begin{align*}
P_{k} = P_0 (N-1)^k \prod_{j=1}^{k}\frac{\binom{j(M+1)-1}{M}\binom{K-M-1}{j(M+1)}}{\binom{K-M-1}{(j-1)(M+1)}\binom{K-1-j(M+1)}{M}},
\end{align*} which can be further simplified to 
\begin{align}\label{eq:Cond1sim}
P_k = P_0 (N-1)^k \prod_{j=1}^{k} \frac{K/(M+1)-j}{j}
\end{align} for all $k\in [g-2]$.

Similarly, by equating~\eqref{eq:PWin} and~\eqref{eq:PWnin2} and simplifying, when $|\mathrm{supp}(\mathrm{v}^{*})|=(g-1)(M+1)$, ${\mathbb{P}(\mathbf{Q}_n=\mathrm{v}^{*}|\mathbf{W}=\mathrm{W}^{*})}$ is constant for all $\mathrm{W}^{*}\in [K]$, 
if 
\begin{align}\label{eq:Cond2}
P_{g-1} = P_{g-2}(N-1)\frac{\binom{M}{g(M+1)-K}\binom{(g-1)(M+1)-1}{M}}{\binom{(g-1)(M+1)}{g(M+1)-K}\binom{K-M-1}{(g-2)(M+1)}}.
\end{align}
By substituting $P_{g-2}$ in terms of $P_0$ using~\eqref{eq:Cond1sim} and further simplifying,~\eqref{eq:Cond2} can be written as  
\begin{align}\label{eq:Cond2sim}
P_{g-1} = P_0 (N-1)^{g-1} \prod_{j=1}^{g-1} \frac{K/(M+1)-j}{j}.
\end{align}

By~\eqref{eq:Cond1sim} and~\eqref{eq:Cond2sim}, it follows that ${\mathbb{P}(\mathbf{Q}_n=\mathrm{v}^{*}|\mathbf{W}=\mathrm{W}^{*})}$ remains constant for all $\mathrm{W}^{*}\in [K]$ if 
\begin{align}\label{eq:Conds12sim}
P_k = P_0 (N-1)^k \prod_{j=1}^{k} \frac{K/(M+1)-j}{j}
\end{align} for all $k\in [g-1]$. 
Defining $r_k = \prod_{j=1}^{k} \frac{K/(M+1)-j}{j}$, 
%as in~\eqref{eq:rk}, 
we can simply rewrite~\eqref{eq:Conds12sim} as 
\begin{align}\label{eq:Pk}
P_k=P_0(N-1)^k r_k
\end{align}
for all $k\in [g-1]$. 
Since $\sum_{k=0}^{g-1} P_k=1$, it follows that 
\begin{align*}\label{eq:P0}
P_0 = \left(1+\sum_{k=1}^{g-1} r_k (N-1)^k \right)^{-1},
\end{align*} which matches our choice of $P_0$ in the proposed scheme. 
In addition, $P_1,\dots,P_{g-1}$ are determined by~\eqref{eq:Pk}, coinciding with our choice of $P_1,\dots,P_{g-1}$ in the proposed scheme. 
This completes the proof of privacy.

\subsection{Proof of Recoverability}\label{subsec:Recoverability}
For simplicity, 
we define ${\mathrm{Y}}_n:=\tilde{\mathrm{Y}}_{\pi^{-1}(n)}$ for all ${n\in [N]}$,
where ${\pi^{-1}: [N]\rightarrow [N]}$ is the inverse of the permutation $\pi$ chosen in Step~1. 
Note that if $\mathrm{v}_n$ is nonzero, ${{\mathrm{Y}}_n = \sum_{i\in \mathrm{supp}(\mathrm{v}_{n})} \mathrm{X}_{i,\mathrm{v}_n(i)}}$, and 
if $\mathrm{v}_n$ is all-zero, ${\mathrm{Y}}_n = 0$.

Given the answers ${\mathrm{A}^{[\mathrm{W},\mathrm{S}]}_1,\dots,\mathrm{A}^{[\mathrm{W},\mathrm{S}]}_N}$, i.e.,  ${\tilde{\mathrm{Y}}_1,\dots,\tilde{\mathrm{Y}}_N}$, or equivalently, ${\mathrm{Y}_1,\dots,\mathrm{Y}_N}$, 
the user can follow a specific procedure described below---which depends on the value of $I$ selected in Step~1---to recover all sub-packets ${\mathrm{X}_{\mathrm{W},1},\dots,\mathrm{X}_{\mathrm{W},N-1}}$ of the demand message $\mathrm{X}_{\mathrm{W}}$. 
The procedure for recovering $\mathrm{X}_{\mathrm{W},j}$ for each $j\in [N-1]$ is as follows:

\setcounter{subsubsection}{0}

\vspace{0.125cm}
\subsubsection{$I\neq g-1$}
By construction, ${\mathrm{v}_{j+1}-\mathrm{v}_1  = \mathrm{b}+\mathrm{c}_j}$. 
This implies that \[{{\mathrm{Y}}_{j+1}-{\mathrm{Y}}_{1} = \sum_{i\in \mathrm{supp}(\mathrm{b})} \mathrm{X}_{i,\mathrm{b}(i)}+\mathrm{X}_{\mathrm{W},j}}.\]
Given the answers of the servers, the user knows
${\mathrm{Y}}_1$ and ${\mathrm{Y}}_{j+1}$. 
Moreover, $\sum_{i\in \mathrm{supp}(\mathrm{b})} \mathrm{X}_{i,\mathrm{b}(i)}$ depends only on the side information messages 
because $\mathrm{supp}(\mathrm{b}) = \mathrm{S}$, and hence, $\sum_{i\in \mathrm{supp}(\mathrm{b})} \mathrm{X}_{i,\mathrm{b}(i)}$ is known to the user. 
Thus, 
$\mathrm{X}_{\mathrm{W},j}$ can be recovered as \[{\mathrm{X}_{\mathrm{W},j}={\mathrm{Y}}_{j+1}-{\mathrm{Y}}_{1} - \sum_{i\in \mathrm{supp}(\mathrm{b})} \mathrm{X}_{i,\mathrm{b}(i)}}.\] 

\vspace{0.125cm}
\subsubsection{$I=g-1$}
It readily follows from the construction that $\mathrm{v}_{j+1}-\mathrm{v}_1 = \mathrm{b}_2+\mathrm{c}_{j}$. 
This further implies that 
\[{{\mathrm{Y}}_{j+1}-{\mathrm{Y}}_{1} = \sum_{i\in \mathrm{supp}(\mathrm{b}_2)} \mathrm{X}_{i,\mathrm{b}_2(i)}+\mathrm{X}_{\mathrm{W},j}}.\] 
Since the user knows ${\mathrm{Y}}_1$ and ${\mathrm{Y}}_{j+1}$ through the servers' answers, and 
$\sum_{i\in \mathrm{supp}(\mathrm{b}_2)} \mathrm{X}_{i,\mathrm{b}_2(i)}$ is known to the user because 
$\mathrm{supp}(\mathrm{b}_2)\subseteq \mathrm{S}$, 
it follows that $\mathrm{X}_{\mathrm{W},j}$ can be recovered as \[{\mathrm{X}_{\mathrm{W},j}={\mathrm{Y}}_{j+1}-{\mathrm{Y}}_{1} - \sum_{i\in \mathrm{supp}(\mathrm{b}_2)} \mathrm{X}_{i,\mathrm{b}_2(i)}}.\]
This completes the proof of recoverability.

\subsection{Proof of Achievable Rate}\label{subsec:Rate}
By definition, the rate is defined as the ratio of $H(\mathbf{X}_{\mathrm{W}})$ to ${\sum_{n=1}^{N} H(\mathbf{A}^{[\mathrm{W},\mathrm{S}]}_n|\mathbf{Q}^{[\mathrm{W},\mathrm{S}]}_n)}$. 
Since each message is a uniform random variable over $\mathbbmss{F}_q^{N-1}$, it follows that
$H(\mathbf{X}_{\mathrm{W}}) = (N-1)\log_2 q=B$. 
To compute ${\sum_{n=1}^{N} H(\mathbf{A}^{[\mathrm{W},\mathrm{S}]}_n|\mathbf{Q}^{[\mathrm{W},\mathrm{S}]}_n)}$, we proceed as follows. 
Fix an arbitrary ${n\in [N]}$. 
In the proposed scheme, $\mathrm{Q}^{[\mathrm{W},\mathrm{S}]}_n = \tilde{\mathrm{v}}_n$, 
and 
$\mathrm{A}^{[\mathrm{W},\mathrm{S}]}_n = \tilde{\mathrm{Y}}_n=\sum_{i\in \mathrm{supp}(\tilde{\mathrm{v}}_n)} \mathrm{X}_{i,\tilde{\mathrm{v}}_n(i)}$. 
If the vector $\tilde{\mathrm{v}}_n$ is all-zero, 
server $n$ does not send back any answer to the user, and hence, 
${H(\mathbf{A}^{[\mathrm{W},\mathrm{S}]}_n|\mathbf{Q}^{[\mathrm{W},\mathrm{S}]}_n=\mathrm{Q}^{[\mathrm{W},\mathrm{S}]}_n)=0}$;
otherwise, if the vector $\tilde{\mathrm{v}}_n$ is nonzero, server $n$ sends back a nontrivial linear combination of the message sub-packets, and hence, 
${H(\mathbf{A}^{[\mathrm{W},\mathrm{S}]}_n|\mathbf{Q}^{[\mathrm{W},\mathrm{S}]}_n = \mathrm{Q}^{[\mathrm{W},\mathrm{S}]}_n)=\log_2 q= B/(N-1)}$. 
This is because the message sub-packets  are independent uniform random variables over $\mathbbmss{F}_q$, 
and 
any nontrivial linear combination of them  is a uniform random variable over $\mathbbmss{F}_q$.
Thus, \[H(\mathbf{A}^{[\mathrm{W},\mathrm{S}]}_n|\mathbf{Q}^{[\mathrm{W},\mathrm{S}]}_n)={(1-\mathbb{P}(\tilde{\mathrm{v}}_n= 0))B/(N-1)}.\] 
Since $\mathbb{P}(\tilde{\mathrm{v}}_n= 0) = P_0/N$, it follows that 
\[H(\mathbf{A}^{[\mathrm{W},\mathrm{S}]}_n|\mathbf{Q}^{[\mathrm{W},\mathrm{S}]}_n)={(1-P_0/N)B/(N-1)}.\] 
By symmetry, we also have \[H(\mathbf{A}^{[\mathrm{W},\mathrm{S}]}_1|\mathbf{Q}^{[\mathrm{W},\mathrm{S}]}_1)=\cdots=H(\mathbf{A}^{[\mathrm{W},\mathrm{S}]}_N|\mathbf{Q}^{[\mathrm{W},\mathrm{S}]}_N).\] 
Thus, \[\sum_{n=1}^{N} H(\mathbf{A}^{[\mathrm{W},\mathrm{S}]}_n|\mathbf{Q}^{[\mathrm{W},\mathrm{S}]}_n) = (N-P_0)B/(N-1),\] which yields the rate $R = (N-1)/(N-P_0)$. 
Substituting $P_0 = (1+\sum_{k=1}^{g-1} r_k (N-1)^k )^{-1}$, we observe that the rate of the proposed scheme is given by
\[R = (N-1)\left(N- \left(1+\sum_{k=1}^{g-1} r_k (N-1)^k \right)^{-1}\right)^{-1},\] which matches the rate $R$ defined in~\eqref{eq:R}. 
This completes the proof of the achievable rate. 

\end{document}